\documentclass[envcountreset,oribibl]{llncs}
\usepackage{amsmath}
\usepackage{amssymb}
\usepackage{amsfonts}
\usepackage{graphicx}
\usepackage[ruled,linesnumbered,noline]{algorithm2e}
\usepackage{bm}

\graphicspath{{./}{images/}}

\title{Network construction: A learning framework through localizing principal eigenvector}

\author{Priodyuti Pradhan\inst{1} \and Sarika Jalan\inst{2} }
\institute{1. Department of Mathematics, Bar-Ilan University, Ramat Gan 5290002, Israel\\
2. Complex Systems Lab, Department of Physics, Indian Institute of Technology Indore, Khandwa Road, Simrol, Indore-453552, India
}

\begin{document}

\maketitle

\begin{abstract}
Information of localization properties of eigenvectors of the complex network has applicability in many different areas which include networks centrality measures, spectral partitioning, development of approximation algorithms, and disease spreading phenomenon. For linear dynamical process localization of principal eigenvector (PEV) of adjacency matrices infers condensation of the information in the smaller section of the network. For a network, an eigenvector is said to be localized when most of its components are near to zero with few taking very high values. Here, we provide three different random-sampling-based algorithms which, by using the edge rewiring method, can evolve a random network having a delocalized PEV to a network having a highly localized PEV. In other words, we develop a learning framework to explore the localization of PEV through a random sampling-based optimization method. We discuss the drawbacks and advantages of these algorithms.  Additionally, we show that the construction of such networks corresponding to the highly localized PEV is a non-convex optimization problem when the objective function is the inverse participation ratio. This framework is also relevant to construct a network structure for other lower-order eigenvectors.
\end{abstract}

\section{Introduction}
Networks furnish a mathematical framework to model and decipher the collective behavior of complex real-world systems. Scrutiny of principal eigenvector (PEV) and the corresponding eigenvalue of the networks are known to provide an understanding of various local and global structural properties as well as the time evolution of dynamics on those networks 
\cite{pevec_nat_phys_2013}, \cite{castellano_localization_2017},  \cite{evec_localization_2017}, \cite{evec_localization_2020}, \cite{evec_loc_jaccard_2021}, \cite{neurons_localization_2014}. Further, different networks' eigenvector-based centrality measures have been proposed to understand the importance of the nodes forming those networks. For example, eigenvector centrality or Katz centrality provides a ranking to nodes of networks based on the entries of the PEV \cite{evec_loc_temporal_network_2017}, \cite{newman_book_2010}. Similarly, the PageRank algorithm which is based on the PEV of Google matrices predicts the importance of the web pages \cite{newman_book_2010}. Also, variants of principal component analysis and independent component analysis have led to radical developments in machine learning approaches \cite{machine_learning_2009}. Furthermore, conditions under which the degree vector of a network and PEV are correlated have been derived leading to the usage of degree vector instead of PEV to approximate various network analysis results \cite{deg_pev_corr_2012}. Furthermore, a network construction for which the message passing equations are exact have been explained and solutions near the critical point in terms of the PEV components have been analyzed \cite{message_passing_2017}. Recently, sensitivity in the network dynamics has been explored using networks' eigenvectors \cite{sensitivity_dynamics_2017}.
Particularly, development of community detection techniques based on the localization properties of eigenvector \cite{community_2010},\cite{community_2015},\cite{mux_community_2017} have been another significant contribution to the complex network analysis. Further, localization has an important role in quantum physics \cite{anderson_loc_1985}, mathematics \cite{loc_math_1_2010}, \cite{loc_math_2_2013}, \cite{loc_math_3_2003}, approximate algorithm development \cite{approx_algo_2015}, machine learning \cite{machine_learning_loc_2016}, numerical linear algebra, matrix gluing, structural engineering, computational quantum chemistry \cite{loc_invariant_subspace_2011}, \cite{anderson_loc_linear_alg_1999}, and in quantum information theory \cite{localization_in_mat_2016}. From real-world systems, it is also evident that underlying interaction patterns form the infrastructure for their different emerging  dynamical responses. For instance, information or rumor propagates through the Facebook-Twitter networks; in brain, the neurons interact to perform specific functions over the underlying network. Moreover, reconfiguration or rewiring of the functional brain networks are required during the learning phases \cite{dynamic_reconfig_2011}. Therefore a scrutiny of the network architecture is important as `underlying structure has crucial  impact on its function' and vice-versa \cite{rev_Strogatz_2001}.

Due to the versatile applications of the eigenvector properties, we analyze the network architecture which optimizes a specific behavior of its PEV. In this article, we study the network structure from a different point of view. Instead of analyzing the properties of a network, we construct a sequence of networks, $\{\mathcal{G}_1, \mathcal{G}_2, \ldots, \mathcal{G}_{final} \}$ optimizing a few specific behaviors of the PEV. The primary aim of this framework is to examine the sequence of networks and learn the network properties collectively when optimizing a function $\zeta:\Re^n \rightarrow \Re$ on the eigenvectors. In other words, we can represent an undirected network by an adjacency matrix which encodes the interactions or relations among $n$ objects (nodes) of a real-world complex system. We consider that the adjacency matrix is symmetric, and hence from the eigenvalue equation, we have $n$ number of eigenvectors which represent $n$ different solutions of the system. Each eigenvector has a different meaning corresponding to the underlying system. We have an interest in network architecture that satisfies few particular behaviors of the eigenvector (solution). The entry of an eigenvector for a symmetric matrix may contain negative, zero, or positive values. Can we tune the eigenvector entries and construct the network structure accordingly$?$ The question one can ask that how values of the eigenvector and the corresponding network structure are related. Reversely, to get a particular behavior for the eigenvector what will be the interaction matrix which in our case is the network's adjacency matrix. The tuning can be performed based on a particular function $\zeta$, and for our purpose here it is the inverse participation ratio (IPR) of the PEV (Eq. (\ref{eq_IPR})). 

It has been demonstrated that few structural features of networks, such as the presence of a hub node, the existence of dense subgraphs, or a power-law degree distribution may lead to the localization of the PEV \cite{MZN2014},\cite{Goltsev_prl2012}. However, the questions which still need investigation and dealt with in this article are; 
\begin{itemize}
\item[a.] Starting with a network having a delocalized PEV, how can one gradually localize the PEV behavior and construct the corresponding network structures?
\item[b.] What is the particular architecture of the optimized network corresponding to a highly localized PEV?
\end{itemize}
Here, we develop a general framework to evolve a  network structure based on rewiring its edges as the PEV of the corresponding adjacency matrix goes from the delocalized to the highly localized state. We devise three different algorithms which use random edge sampling (Hub-based, Monte-Carlo-based, and simulated annealing based) to find out the final network structure. Moreover, we show that the optimized network concerning the highly localized PEV has a distinctive architecture.

The article is organized as follows: Section 2 motivates to study PEV localization. Section 3, contains the notations and definitions used in the later discussion. Also, it includes a brief explanation and formulation of our work. Section 4, illustrates various algorithms on edge rewiring-based optimization in detail. Finally, section 5, summarizes the current study and discusses the open problems for further investigations.

\section{Motivation}
We know relations or interactions are everywhere either interactions of power grid generators to provide proper functioning of the power supply over a country, or interactions of bio-molecules inside cell to proper functioning of cellular activity or interactions of neurons inside brains to perform specific functions or interactions among satellites to provide accurate GPS services or interactions among quantum particle to form quantum communications or recent coronavirus spread \cite{rev_Strogatz_2001}, \cite{dynamic_reconfig_2011}, \cite{quantum_internet}, \cite{GPS_network}, \cite{COVID-19}. For all of them there are two things that are common –- network structure and interactions mechanism among the components or dynamics running on top of that. For instance, for a linear dynamical model (Fig. \ref{schematic_pev_lov}) 
\begin{figure}[t]
\begin{center}
\includegraphics[width=3.2in, height=2in]{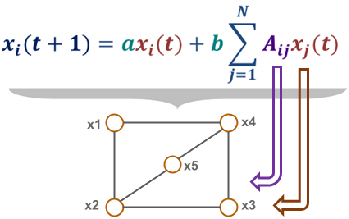}
\caption{A linear dynamical model. Here, $x_i$ captures some information on node $i$ and $A_{ij}$'s are the structure of the network and $a$ and $b$ are the model parameters. We seek to understand the long term behavior of the information exchange pattern.}
\label{schematic_pev_lov}
\end{center}
\end{figure}
\begin{equation}\label{power_iteration}
\bm{x}(t+1) =  {\bf M}^{t+1} \bm{x(0)}\overset{t\rightarrow \infty}{\Longrightarrow} \bm{x}^{*} \sim  \bm{v}_1^{\bf M}  
\end{equation}
where ${\bf M}=a{\bf I}+b{\bf A}$ is the transition matrix, $\bm{x(0)}$ is the initial state of the system, $\bm{x}^{*}$ is the stationary state and $\bm{v}_1^{\bf M}$ is the PEV of {\bf M}. To find the pattern or understand the behavior of the linear dynamics on networks, we iterate the matrix multiplication process for a large number of times or we can find out the PEV ($\bm{v}_1^{\bf M}$) of {\bf M} (Eq. (\ref{power_iteration})). There is another way, we can also solve the problem. We can see that eigenvectors of {\bf M} are the same as eigenvectors of {\bf A} \cite{pevec_nat_phys_2013}.
\begin{equation}\nonumber
{\bf M}\bm{v}_i^{\bf A}=a{\bf I}\bm{v}_i^{\bf A}+b{\bf A}\bm{v}_i^{\bf A}
    =[a+b\lambda_i^{\bf A}]\bm{v}_i^{\bf A}
    =\lambda_i^{\bf M}\bm{v}_i^{\bf A}
\end{equation}
where {\bf I} is the identity matrix and $\lambda_i^{\bf M}=a+b\lambda_i^{\bf A}$ and. Thus, $\bm{x}^{*} \sim  \bm{v}_1^{\bf M}=\bm{v}_1^{\bf A}$. {\em Hence, understanding the behavior of the PEV of the adjacency matrix is enough to understand the information flow pattern for linear dynamical systems.} Therefore, the network properties which enhance PEV localization of adjacency matrices can implicitly restrict the linear-dynamics in a smaller section of the network. A fundamental question at the core of the structural-dynamical relation is that: How can we find network structure having highly localized PEV? Here, we provide random-sampling based algorithm to construct localized network structure. 

\section{Problem formulation}
We represent a graph, $\mathcal{G} =\{V,E\}$, where $V=\{v_1, v_2,\ldots,v_n\}$ is the set of vertices and $E=\{e_1, e_2,\ldots,e_{m}|e_{m}=(v_i, v_j)\} \subseteq U$ is the set of edges. We define the universal set $U = V \times V=\{(v_i, v_j)| v_i, v_j \in V\; \text{and}\; i \neq j\}$ which contains all possible ordered pairs of vertices excluding the self-loops and the complementary set can be defined as $E^c = U - E=\{(v_i, v_j)| (v_i, v_j) \in U\; \text{and}\; (v_i, v_j) \notin E\}$ i.e., $E\cap E^c=\varnothing $ and $E \cup E^c= U$. We denote the adjacency matrices corresponding to $\mathcal{G}$  as ${\bf A} \in \Re^{n \times n}$ and which can be defined as $(a)_{ij} = 1$, if  $v_i \sim v_j$ and $0$ otherwise.
The degree of a node can be represented as $d(v_i)=\sum_{j=1}^n a_{ij}$, and the average degree of $\mathcal{G}$ can be defined as $\langle k \rangle = \frac{1}{n}\sum_{i=1}^n d(v_i)$. Here, we consider $|V|=n$, $|E|=m$, and $|E^c|=\frac{n(n-1)}{2}-m$.  
The \textit{spectrum} of $\mathcal{G}$ is the set of the eigenvalues $\{\lambda_1, \lambda_2, \ldots, \lambda_n\}$ of ${\bf A}$. Without loss of generality we can order the eigenvalues of ${\bf A}$ as $\lambda_1 > \lambda_2 \geq \cdots \geq \lambda_n$ and corresponding eigenvectors as $\bm{x}_1, \bm{x}_2, \cdots, \bm{x}_n$ respectively. Here, ${\bf A}$ is a real symmetric matrix, and each has real eigenvalues. In addition, the networks are connected. Hence, we know from the Perron-Frobenius theorem \cite{miegham_book2011} that all the entries in the PEV ($\bm{x}_1$) of ${\bf A}$ are positive. We quantify localization of the PEV through the inverse participation ratio, which is the sum of fourth power of the eigenvector entries and calculate \cite{Goltsev_prl2012}, \cite{evec_localization_2017} as follows:
\begin{equation} \label{eq_IPR}
Y_{\bm{x}_{1}}=\sum_{i=1}^n (x_{1})_{i}^4 
\end{equation}
where $(x_{1})_{i}$ is the $ith$ component of $\bm{x}_1$ and $||\bm{x}_1||_2^2=1$. A delocalized eigenvector with component $(\frac{1}{\sqrt{n}},\frac{1}{\sqrt{n}},\ldots,\frac{1}{\sqrt{n}})$ has $Y_{ \bm{x}_1}=\frac{1}{n}$, whereas the most localized eigenvector with components $(1,0,\ldots,0)$ yields an IPR value equal to $Y_{\bm{x}_1} = 1$. A network is said to be regular if each node has the same degree \cite{miegham_book2011}. It also turns out that for any regular graph (Theorem 6 \cite{miegham_book2011}), we get PEV, $\bm{ x}_1=(\frac{1}{\sqrt{n}},\frac{1}{\sqrt{n}},\ldots,\frac{1}{\sqrt{n}})$. Hence, $Y_{\bm{ x}_1}=\frac{1}{n}$, corresponds to the most delocalized PEV. Therefore, for any regular network IPR value of the PEV provides the lower bound. Hence, a sparse as well as a dense regular network contains delocalized PEV. Now, we can consider a disconnected graphs where each node is isolated from each other and each node has a self-loop. The adjacency matrix can be represented with the $n \times n$ identity matrix. For this disconnected networks, $Y_{\bm{ x}_1} = 1$. In another situation if we consider only $n$ number of isolated nodes with $|E|=0$. We have a zero matrix and for which we can choose $\bm{ x}_1=(1,0,\ldots,0)$ and $Y_{\bm {x}_1} = 1$. These are the special cases. Additionally, for any disconnected network with less than $n$ number of components, the PEV entries might be zeros. Hence, for a connected network, IPR value lies between $ 1/n \leq Y_{\bm{ x}_{1}}<1$. Therefore, it is evident that finding out a network architecture for a given $n$ with delocalized PEV is easier than searching for a connected network structure with highly localized PEV. For a given $n$ and $m$, our aim is to get a connected network
which has the most localized PEV. In other words, we can state the problem as, we search for a binary symmetric matrix (${\bf A}_{final}$) which is irreducible and which has the PEV with maximum IPR value. Also we have an interest to know the properties of sequence of the adjacency matrices $\{{\bf A}_1, {\bf A}_2, \ldots, {\bf A}_{final} \}$ during the searching process which can maximize the IPR value. The search space for a given $n$ and $m$ is of the order $O(n^{2m})$ \cite{search_space}. Therefore, we use random-sampling to formulate this problem.
\begin{figure}[t]
\begin{center}
\includegraphics[width=3.8in, height=2.4in]{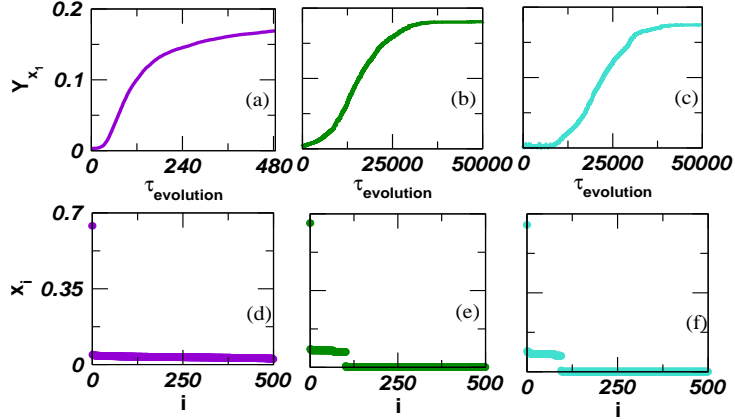}
\caption{Changes of the IPR value during the random-sampling process (a) hub-based algorithm (b) MC-based algorithm and (c) SA-based algorithm. Sorted PEV entries of the final network obtained from (d) hub-based (e) MC based and (f) SA based algorithm has been portrayed. Network size: $n=500$, $\langle k \rangle =10$, and $\tau_{evolution}$ counts the number of edge rewirings.}
\label{opt_ipr}
\end{center}
\end{figure}

\section{Methodology and Results}
We use randomized algorithms based on edge rewiring to construct a network architecture that corresponds to a highly localized PEV in an iterative manner. It is very natural that modification in the entries of an adjacency matrix leads to a change in the spectral properties (eigenvalues and eigenvectors) which also simultaneously change the network architecture. We use this fact to develop randomized algorithms. The modification in the adjacency matrix can be performed by removing or adding edges as well as nodes or rearrangements of the edges in $\mathcal{G}$ \cite{evec_localization_2017}. Here, we devise algorithms by rearrangement of the edges to get highly localized PEV when $\mathcal{G}$ remains connected, and the number of nodes and edges remain fixed. In the following, we discuss the algorithms in detail.

\subsection{Hub-based algorithm}
It is well known that networks with localized PEV have a hub node \cite{MZN2014}. Here, we attempt to connect PEV localization and the corresponding network structure. We use the presence of a hub node heuristic to develop a randomized algorithm. It iteratively forms a hub node starting from an Erd\"os-R\'enyi (ER) random network and records the IPR value as well as stores the sequence of networks $\{\mathcal{G}_{ER},\mathcal{G}_{2},\ldots,\mathcal{G}_{final}\}$. Starting the algorithm, with an ER random network is an artifact as it provides the delocalized PEV \cite{deloc_pev}. The initial ER random network, $\mathcal{G}(n,p)$ is generated with an edge probability $p=\frac{\langle k \rangle}{n}$. Without loss of generality, we consider $v_1$ will be the hub node at the end of the iterative process. We select an edge $e_r \in V-\{v_1\} \times V-\{v_1\}$ uniformly at random from $\mathcal{G}$, and remove it. Simultaneously, add it between $v_1$ to $v_k$, if $(v_1,v_k) \notin E$. We repeat the process until $v_1$ connects to all the remaining nodes and becomes the hub node (algorithms \ref{hub}). At the end of the iterative process we get $d(v_1)=n-1$ for $\mathcal{G}_{final}$.
This random iterative hub formation algorithm keeps unchanged the network size. We can see the IPR value during the evolution from Fig. \ref{opt_ipr}(a). Interestingly, it shows the changes in the IPR value for the sequence of networks $\{\mathcal{G}_{ER},\mathcal{G}_{2},\ldots,\mathcal{G}_{final}\}$ collectively. Moreover, we depict the sorted PEV entry values in Fig. \ref{opt_ipr}(d), which indicates the magnitude of the maximum PEV entry value is much larger than the rest of the entry values. The question will arise whether the IPR value of $\mathcal{G}_{final}$ is close to the optimal. We use some results from previous research on the upper bound on the maximal entry value of the PEV for a connected network to make a possible conclusion about the optimality of our results. The maximal PEV entry value can be obtained for the star network and it is $1/\sqrt{2}$ and all other entries are same which are $1/\sqrt{2(n-1)}$ \cite{bound_pev_entries_2000}. The maximum PEV entry value obtains from the $\mathcal{G}_{final}$ is close to $1/\sqrt{2}$. This simple hub node formation-based algorithm works well and is easier to implement to get networks with highly localized PEV. Later on, we devise other algorithms which provide better results than algorithm \ref{hub}. We use the C++ language and STL library to implement all the algorithms.
\begin{algorithm}[thb]
\label{hub}
\SetCommentSty{sf}
\caption{Hub-based($n$, $\langle k \rangle$)}
\Indp 
${\bf A} \leftarrow \mathcal{G}(n,p)$\\
$Y_{\bf x_1} \leftarrow \zeta(\bf x_1)$ \\ 
 \While{$k=2$ to $n$}
{
  choose an edge $(v_i,v_j) \in V-\{v_1\} \times V-\{v_1\}$ uniformly at random\\
  \If{$(v_1,v_k) \notin  E$}
  {
  remove $(v_i,v_j)$ and add as $(v_1,v_k)$ in $\mathcal{G}$ and denotes it as $\mathcal{G}^{'}$\\
  if $\mathcal{G^{'}}$ is not connected then goto step 4\\
  ${\bf A} \leftarrow {\bf A^{'}}$ \\
  $Y_{\bf x_1} \leftarrow \zeta(\bf x_1)$
  }
  store $Y_{\bf x_1}$ and ${\bf A}$
}
\Indm
\end{algorithm}
To find out the eigenvector, we bind the LAPACK routine \emph{ssyevr} with C++ code.
\begin{figure}[t]
\begin{center}
\includegraphics[width=3.1in, height=1in]{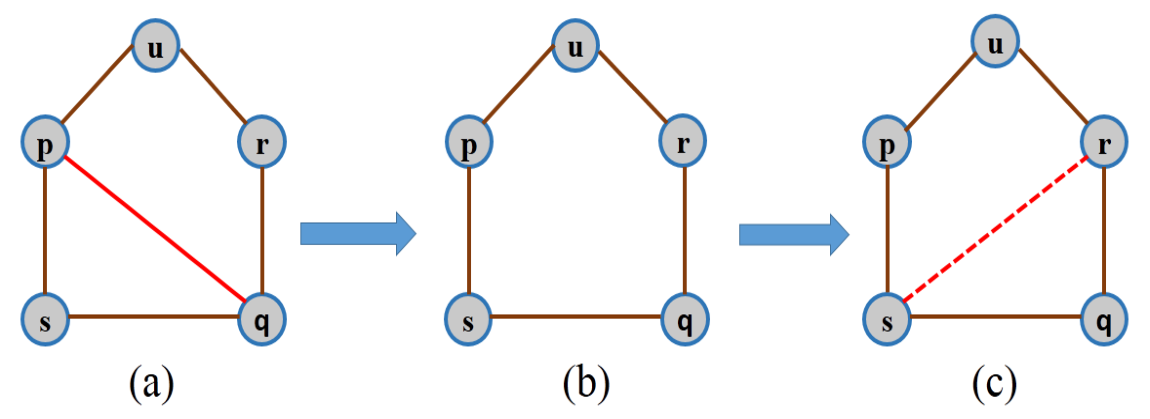}
\caption{Single edge rewiring. Select an edge randomly between a node pair $(p,q)$ and remove it. Add the edge between a randomly selected node pair $(r,s)$.}
\label{rewire2}
\end{center}
\end{figure}
% 
% We run the algorithm for several realizations with varying network size and obtain maximal PEV entry value close to $\frac{1}{\sqrt{2}}$ but rest of the entry values are almost same.  
%
%In other words, we can state the problem as, we search for a binary symmetric matrix ${\bf A_{opt}}$ which is irreducible and which has the PEV with maximum IPR value. Also we have an interest to know the {\bf properties of} sequence of the adjacency matrix $\{\bf A_1, A_2, \ldots, A_{opt} \}$ during the searching process which can maximize the IPR value if we start from an initial matrix ${\bf A_1}$. 
%
%
%To get an adjacency matrix corresponds to highly localized PEV and the sequence $\{\bf A_1, A_2, \ldots, A_{opt} \}$ which lead to the maximum IPR value, we have borrowed the evolution of networks with edge rewiring and used optimization on top of that. 

\subsection{Monte-Carlo based algorithm}
In the previous algorithm, we select an edge at random from the set $V-\{v_1\} \times V-\{v_1\}$ and always add it to the node $v_1$. However, if we make the position of the edge removal and addition (or we can say edge rewiring) more flexible, and accept those edge rewirings which can improve the IPR value, we get an impressive result. In particular, we achieve significant improvement as well as in the network structure than the Hub-based algorithm. We formulate an optimization problem as: given a connected network $\mathcal{G}$ with $n$ vertices, $m$ edges and a function $\zeta:\Re^{n} \rightarrow \Re$, we want to compute the maximum possible value of an objective function $\zeta(\bm{ x}_1)=\sum_{i=1}^n (x_1)_i^4$ over all the simple, connected, and undirected network $\mathcal{G}$. The optimization problem can be written as finding an irreducible binary symmetric matrix ${\bf A}$, for which $\sum_{i=1}^n (x_{1})_{i}^4$ will be maximum such that ${\bf A} \bm{x}_1=\lambda_1 \bm{x}_1$, $||\bm{ x}_1||^2_2=1$ and $(x_1)_i>0$, $i \in \{1,2,\ldots,n\}$ (detail about the objective function in Appendix). The first constraint simply says that $\bm{x}_1$ is the PEV of a symmetric matrix ${\bf A}$ and it is in $l_2$ norm. The second constraint implicitly stipulates that the network must be connected (from the Perron-Frobenius theorem). We refer the initial network as $\mathcal{G}_{init}$ and the optimized network as $\mathcal{G}_{final}$. The network evolution emerges sequence of networks as $\{\mathcal{G}_{init}, \mathcal{G}_1, \mathcal{G}_2, \ldots, \mathcal{G}_{final} \}$.

For a single-edge rewiring, we choose an edge $e_i \in E$ uniformly at random from $\mathcal{G}$ and remove it (Fig. \ref{rewire2}). At the same time, we introduce an edge in the $\mathcal{G}$ from $E^c$, which preserves the total number of edges during the network evolution in $\mathcal{G}$. Hence, each edge rewiring is a two-step process, (i) removal of an edge followed by (ii) addition of an edge (Fig. \ref{rewire2}). 
We remark that during the network evolution there is a possibility that an edge rewiring disconnects the network. To avoid this situation, we only approve those rewirings which yield the network connected. To check the connectedness after an edge rewiring, we use a depth-first search (DFS) algorithm 
\begin{figure}[t]
\begin{center}
\includegraphics[width=2.6in, height=1.8in]{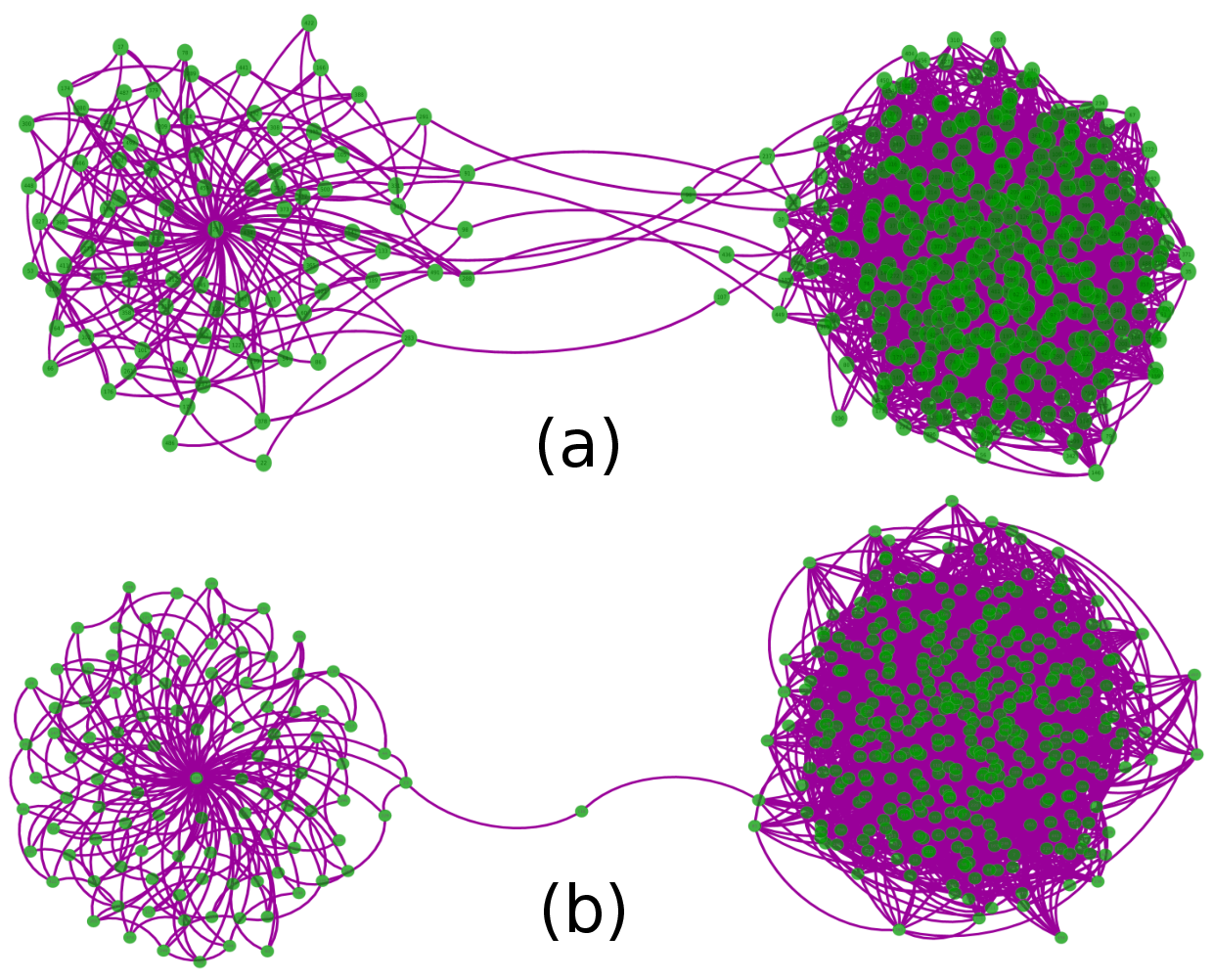}
\caption{(a) Intermediate network structure during optimization (b) Optimized structure obtained at the end of the iteration process with $Y_{x_1} \approx 0.20$; $n=500$ nodes and $\langle k \rangle =10$ edges.}
\label{opt_net}
\end{center}
\end{figure}
\cite{algorithms_cormen_2009}.
\begin{algorithm}[!htb]
\label{monte_carlo}
\SetCommentSty{sf}
\caption{MC-based-IPR-Optimization($n$, $\langle k \rangle$)}
\Indp
${\bf A} \leftarrow \mathcal{G}(n,p)$\\
$Y_{\bf x_1} \leftarrow \zeta(\bf x_1)$\\ 
\While{$Y_{x_1}$ not saturated}
{
  rewire an edge uniformly at random in $\mathcal{G}$ and denotes it as $\mathcal{G'}$ \\
  if $\mathcal{G^{'}}$ is not connected then goto step 4\\
  $Y_{\bf x_1}^{'} \leftarrow \zeta(\bf x_1^{'})$\\
  \If{$Y_{\bf x_1}^{'} > Y_{\bf x_1}$ }
      {
         ${\bf A} \leftarrow {\bf A^{'}}$  \\
         $Y_{\bf x_1} \leftarrow Y_{\bf x_1}^{'}$
       }
  store $Y_{\bf x_1}$ and ${\bf A^{'}}$
}  
\Indm
\end{algorithm}

The Monte-Carlo (MC) based optimization (in algorithm \ref{monte_carlo}) can be summarized as follows. We find ${\bm{x}_{1}}$ of an ER random graph $\mathcal{G}$ and calculate the IPR value of ${\bm{x}_1}$. We rewire one edge uniformly and independently at random in $\mathcal{G}$ to obtain another graph $\mathcal{G'}$. We check whether $\mathcal{G'}$ is connected, if not the edge rewiring step is repeated till we get another $\mathcal{G'}$ which is a connected network.  We find out the PEV of ${\bf A'}$ matrix and calculate the IPR value of ${\bm{x}_1^{'}}$. We replace ${\bf A}$ with ${\bf A'}$, if $Y_{\bm{x}^{'}_{1}} > Y_{\bm{x}_{1}}$. Steps are repeated until the IPR value gets saturated which corresponds to the optimized network. The recorded value of $Y_{\bm{x}_{1}}$ variable during the optimization process gives an increment in the IPR value which is depicted in Fig \ref{opt_ipr}(b). We depict a network at an intermediate evolution stage and the final optimized one in Fig. \ref{opt_net}. It indicates that the optimized network structure contains two graph components connected via a single node as we have seen in \cite{evec_localization_2017}. The sorted PEV entries obtained from the optimized network in Fig. \ref{opt_ipr}(e) portrays that we are very close to the optimal IPR value. However, there is a difference in the eigenvector entries from Fig. \ref{opt_ipr}(d) and the $\mathcal{G}_{opt}$ obtain from MC based algorithm has maximum degree, $d_{max}<<n-1$. It indicates that optimal IPR value depends on the particular entry value behavior of PEV. Here, we consider ER random network as initial. Now, if we change the initial network instead of ER random network, then there is a chance of failure to the MC method. Interestingly, we have found out one such situation and discussed it in the following using the simulated annealing-based method.
\begin{figure}[t]
\begin{center}
\includegraphics[width=3.2in, height=1.4in]{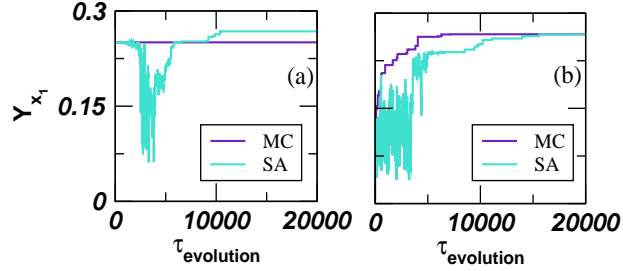}
\caption{Changes of IPR value when initial network as (a) star graph (b) path graph, where $n=500$ nodes.}
\label{ipr_var}
\end{center}
\end{figure}

\begin{algorithm}[!htb]
\label{SA}
\SetCommentSty{sf}
\caption{SA-based-IPR-Optimization($n$, $\langle k \rangle$, $temp$, $\kappa$) }
\Indp 
${\bf A} \leftarrow \mathcal{G}(n,p)$\\
$Y_{\bf x_1} \leftarrow \zeta(\bf x_1)$\\
$temp = 0.9$\\
\While{$Y_{\bf x_1}$ not saturated}
{
  rewire an edge uniformly at random in $\mathcal{G}$ and denotes it as $\mathcal{G'}$ \\
  if $\mathcal{G^{'}}$ is not connected then goto step 5\\
  $Y_{\bf x_1}^{'} \leftarrow \zeta(\bf x_1^{'})$\\
  \If{$Y_{\bf x_1}^{'} > Y_{\bf x_1}$ }
      {
         ${\bf A} \leftarrow {\bf A^{'}}$  \\
         $Y_{\bf x_1} \leftarrow Y_{\bf x_1}^{'}$
       }
   \Else
   {
    pick a random number $r$ from uniform distribution in (0,1) \\
     \If{$r<e^{\frac{(Y_{\bf x_1}^{'}-Y_{\bf x_1})}{(\kappa*temp)}}$}
     { 
        ${\bf A} \leftarrow {\bf A^{'}}$  \\
         $Y_{\bf x_1} \leftarrow Y_{\bf x_1}^{'}$
     }
   }
  store $Y_{\bf x_1}$ and ${\bf A^{'}}$\\
  $temp \leftarrow temp*0.998$
}  
\Indm
\end{algorithm}
\subsection{Simulated annealing based algorithm}
The simulated annealing (SA) is a randomized algorithm widely used in solving optimization problems motivated by the principles of statistical mechanics \cite{simuated_annealing_1983}. The important part of the SA-based algorithm is accepting solutions that satisfy the Gibbs-Boltzmann function $e^{-\mathcal{E}/\kappa*temp}$. In our problem, we consider the objective function to maximize instead of minimizing, so we have made the changes accordingly in the algorithm. We set the initial temperature, $temp=0.9$ and after each iteration decreases it by the cooling schedule $tem=tem*0.98$ and also fix the Boltzmann constant $\kappa$ to $100$. 

\begin{figure}[t]
\begin{center}
\includegraphics[width=3.5in, height=2in]{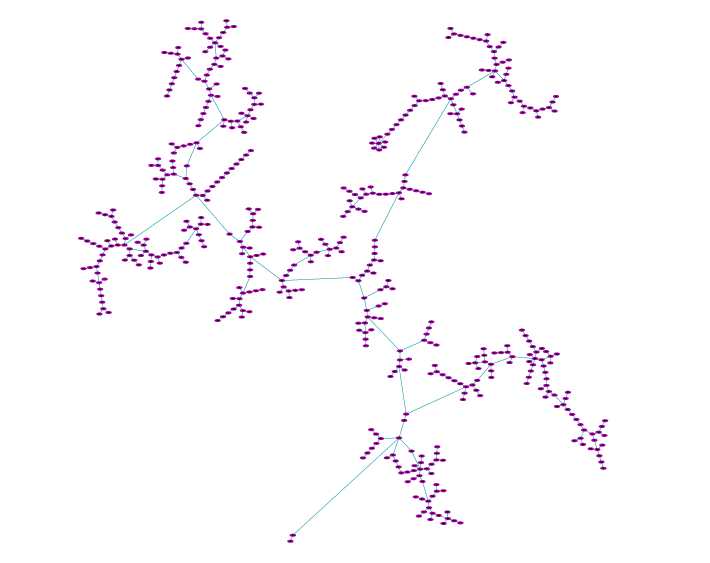}
\caption{Optimized structure obtain through the SA algorithm when initial structure is the star network. The $Y_{\bm{x}_1} \approx 0.267$, $n=500$ nodes and $m=499$ edges.}
\label{opt_net_from_star}
\end{center}
\end{figure}
\begin{figure}[t]
\begin{center}
\includegraphics[width=2.4in, height=0.7in]{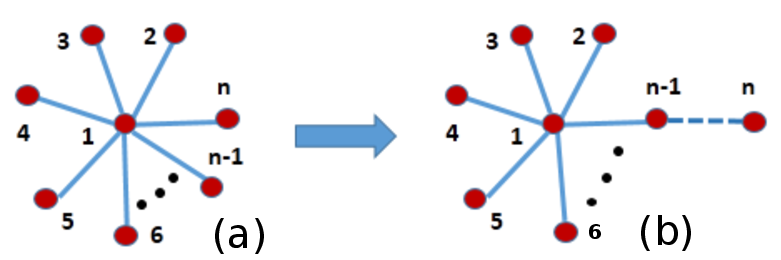}
\caption{MC algorithm with (a) initial star network (b) rewired network configuration.}
\label{star_falure}
\end{center}
\end{figure}
If we consider a star network with $n$ nodes labelled as $\{1,2,\ldots,n\}$ with the hub node being labelled with $1$, then using eigenvalue equation corresponding to $\lambda_{1}$, we get the PEV, 
\begin{equation}\nonumber
{\bm{x}_1}=\biggl(\frac{1}{\sqrt{2}},\frac{1}{\sqrt{2(n-1)}},\ldots, \frac{1}{\sqrt{2(n-1)}}\biggr)    
%\end{equation}
\text{ and }
%\begin{equation}\nonumber
Y_{\bm{x}_1}^{s} = \frac{1}{4} + \frac{1}{4(n-1)}    
\end{equation}
Therefore, when $n \rightarrow \infty$, we get $Y_{\bm{x}_1} \rightarrow \frac{1}{4} \approx 0.25$ for the star network. By looking the PEV entries and the IPR value, it might happen that star network has the most localized PEV, but it is not true.
It becomes clear if we provide a star structure as an initial network to the MC-based algorithm. We do not show any increment in the IPR value, and it sticks to the local maxima (Fig. \ref{ipr_var}(a)). We explain the reason behind the failure in the following. 
After removing an edge connected to the hub node in the star network (Fig. \ref{star_falure}(a)), it must be connected to any peripheral node (Fig. \ref{star_falure}(b)). From the adjacency matrix of the rewired network structure (in Fig. \ref{star_falure}(b)) we solve the eigenvalue equation and find out, 
\begin{equation}\nonumber
{\bm{x}_1}=\biggl((x_1)_1,\frac{1}{\lambda_1}(x_1)_1,\ldots,\frac{1}{\lambda_1}(x_1)_1, \frac{\lambda_1}{\lambda_1^2-1}(x_1)_1,\frac{1}{\lambda_1^2-1}(x_1)_1\biggr)    
\end{equation}
and
\begin{equation}\nonumber
Y_{\bm{x}_1}^{c} = (x_1)_1^4 \biggl(1+\frac{n-3}{\lambda_1^4} + \frac{\lambda_1^4+1}{(\lambda_1^2-1)^4}\biggr)    
\end{equation}
where $(x_1)_1^2=\frac{\lambda_1^4 - \lambda_1^2}{2\lambda_1^4+(n-3)\lambda_1^2-(n-3)}$ and $\lambda_1^2= \frac{(n-1)+\sqrt{(n-1)^2-4(n-3)}}{2}$. We observe numerically that for varying network size upto $n=2000$, $Y_{\bm{x}_1}^{c} < Y_{\bm{x}_1}^{s}$ and after that $Y_{\bm{x}_1}^{c} \approx Y_{\bm{x}_1}^{s}$. Hence, MC based algorithm does not accept the configuration in Fig. \ref{star_falure}(b) and iterates forever without enhancement to the IPR value. However, in case of SA based algorithm, initially, when the $temp$ is high then from uniform distribution the algorithm accepts the configuration as in Fig. \ref{star_falure}(b). It leads to a better optimal value and gives an optimized structure which is different from the star network (Fig. \ref{opt_net_from_star}). Moreover, giving a path network as an initial network to both the MC and SA based method gives an optimized structure as in Fig. \ref{opt_net_from_star} and get an improvement in the IPR value (Fig. \ref{ipr_var}(b)). It indicates that success of the MC algorithm depends on the choice of the initial network. Furthermore, from the numerical simulations, we have learned that when the number of edges $|E| >> n-1$ then the MC and SA based algorithm works well to find out an optimized network structure. 

\section{Conclusion}
We explore different network construction algorithms which optimizing the behavior of PEV. We construct the network structure through the optimization process that possesses highly localized PEV quantified by the IPR value. This approach provides a comprehensive way to investigate not only the optimized network but also intermediate networks before an optimized structure is found. 
%In other words, we develop a learning framework to explore the localization of PEV through a sampling-based optimization method.
Furthermore, we restrict our study to the adjacency matrix of the network. It is also interesting to examine for other matrices related to graphs such as Laplacian matrix, modularity matrix. 
Here, eigenvector behavior has been regulated based on the particular function which is IPR. It is good to define another function that can tune the negative or zeros entries of the eigenvectors or some part of the eigenvector and based on that one can construct a network structure. This framework also helps to construct a network for other lower-order eigenvectors. Low order eigenvectors have also been studied to develop machine learning tools. They used IPR as well as another kind of measure for the eigenvector localization, called statistical leverage scores \cite{low_order_evec_2013}, \cite{leverage_score_2011} which has an impact on statistics and modern big data analysis.

It is also interesting to distribute weights in a weighted network so that PEV becomes localized \cite{Goltsev_prl2012}. Moreover, removing edges make the PEV delocalized, and also adding edges to the network produces delocalized PEV, but it may be possible to add a proper number of edges to the network which forms the highly localized PEV. We have not included the complexity analysis of the algorithms, and it is an exciting part to do in the future. Finally, we devise edge rewiring-based optimization algorithms which allow us to learn about the network structure from the PEV and it may be relevant to be used later to develop machine learning tools.

\noindent {\bf Acknowledgment:} SJ acknowledges DST, Govt. of India grant EMR/2014/000368 for financial support. PP is indebted to Michael C. Grant (StackExchange, CVX Research, Stanford University) and Amit Reza (IIT Gandhinagar, India) for useful discussion on the optimization problem and members of CSL at IIT Indore for discussions.

\section*{Appendix}

Next, we discuss in details about the objective function and constraints. 
% A convex optimization problem requires (a) the objective function must be convex (b) the inequality constraint functions must be convex, and (c) the equality constraint functions must be affine \cite{convex_opt_2009}.
% 
\begin{lemma}
$\zeta({\bm{x}_1})=\sum_{i=1}^n (x_{1})_{i}^4$ is a convex function when $(x_1)_i \in (0,1)$, $i \in \{1, 2, \ldots, n\}$. 
\end{lemma}

\begin{proof}
Convexity of the objective function $\zeta({\bm{x}_1})$ can be examined by employing Hessian test \cite{nonlinear_opt_2014}. One can construct the Hessian matrix from $\zeta({\bm{x}_1})$ and show that it is positive semidefinite.
The partial derivative of $\zeta({\bm{x}_1})$ are given by 
\begin{equation}
 \frac{\partial \zeta({\bm{x}_1})}{\partial (x_1)_{i}}=4(x_1)_i^3, \; i=\{1,2,\ldots,n\} 
\end{equation}
and hence,
\begin{equation}
 \frac{\partial^2 \zeta({\bm{x}_1})}{\partial (x_1)_{i}\partial (x_1)_{j}}=
 \begin{cases}
 12(x_1)_i^2 & \quad i = j \\
 0           & \quad i \neq j\\
\end{cases}
\end{equation}
Now we can write the Hessian matrix as  
\[
  \nabla^2\zeta({\bm{x}_1}) =
   \begin{bmatrix}
    12(x_1)_1^2 & & \text{\LARGE0}\\
    & \ddots & \\
   \text{\LARGE0} & & 12(x_1)_n^2
   \end{bmatrix}
\]
Hessian matrix is positive semidefinite if all the eigenvalues of $\nabla^2\zeta({\bm{x}_1})$ are non-negative. Here it is clear that eigenvalues of $\nabla^2\zeta({\bm{x}_1})$ are $\{12 (x_1)_i^2: i = 1,2, \cdots n\}$ as $\nabla^2\zeta({\bm{x}_1})$ is a diagonal matrix.
Since $(x_1)_i \in (0,1)$, therefore all the eigenvalues of $\nabla^2\zeta({\bm{x}_1})$ are nonnegative, and hence the Hessian matrix is a positive semidefinite matrix. Therefore, the objective function $\zeta({\bm{x}_1})$ is a convex function. 
\end{proof}

\begin{lemma}
 $\mathcal{C}=\{ {\bm{x}_1} \in (0,1)^n|\;||{\bm{x}_1}||^2_2=1\}$ is a non-convex set.
\end{lemma}

\begin{proof}
A set $\mathcal{C} \subseteq \Re^n$ is called convex if for any ${\bm{x}, \bm{y}} \in \mathcal{C}$, ${\bm{x}\neq \bm{y}}$ and any $\theta \in [0, 1]$, the point $\theta \bm{x} + (1-\theta)\bm{y}$ belongs to $\mathcal{C}$ \cite{nonlinear_opt_2014}. 
To validate $\mathcal{C}$ as a non-convex set, any arbitrary point $\bm{z} \in \Re^n $  has been considered and it can be written as a convex combinations of $\bm{x}$ and $\bm{y}$ i.e., $\bm{z}=\theta \bm{x} + (1-\theta)\bm{y}$ by choosing
an arbitrary value of $\theta$. Thus, we have
\begin{equation}
(z_1,z_2,\ldots,z_n)=(\theta x_1+(1-\theta)y_1,\theta x_2+(1-\theta)y_2,\ldots,\theta x_n+(1-\theta)y_n) 
\end{equation}
From the above equation, we get, 
\begin{equation}
\sum_{i=1}^n z_i^2=\theta^2+(1-\theta)^2+2\theta(1-\theta)\sum_{i=1}^n x_iy_i 
\end{equation}
Now, to check the convexity, one has to show that $\bm{z} \in \mathcal{C}$, i.e., $\sum_{i=1}^n z_i^2=1$.
Since ${\bm{x}\neq \bm{y}}$, it gives $\sum_{i=1}^n x_iy_i\neq 1$. Now, for specific $\theta=1/2$ , $\sum_{i=1}^n z_i^2 \neq 1$, this implies that the relation  $\sum_{i=1}^n z_i^2=1$ does not satisfy for any arbitrary value of $\theta$. Hence, $\bm{z} \notin \mathcal{C}$ and therefore, $\mathcal{C}$ is a non-convex set.
\end{proof}
\begin{theorem}
 Considering $\zeta(\bm{x}_1)$ as an objective function, principal eigenvector localization over undirected network is a non-convex optimization problem. 
\end{theorem}  
\begin{proof}
It is notable from Lemma 1 that the objective function $\zeta(\bm{x}_1)$ is a convex function but on the other hand, Lemma 2, says that the constraint, $\mathcal{C}=\{ {\bm{x}_1} \in (0,1)^n|\;||{\bm{x}_1}||^2_2=1\}$ is a non-convex set. By definition, a convex optimization problems consist of minimizing of a convex functions over convex sets, or maximizing a concave functions over convex sets \cite{nonlinear_opt_2014}.
Jointly, conflicting characteristic of constraint and objective function shows that the principal eigenvector localization over simple undirected, and unweighted network is a non-convex optimization problem.
\end{proof}


\begin{thebibliography}{0}

%\bibitem{metastable_loc_2016} R. S. Ferreira, R. A. da Costa, S. N. Dorogovtsev, and J. F. F. Mendes, Metastable localization of disease in complex networks, Phys. Rev. E \textbf{94}, 062305 (2016).

\bibitem{pevec_nat_phys_2013} J. Aguirre, D. Papo, and J. M. Buld\'{u}, Successful strategies for competing networks, Nat. Phys. {\bf 9}, 230 (2013).

\bibitem{castellano_localization_2017} C. Castellano and R. Pastor-Satorras, Topological determinants of complex networks spectral properties: structural and dynamical effects, Phys. Rev. X {\bf 7}, 041024 (2017).

\bibitem{evec_localization_2017} P. Pradhan, A. Yadav, S. K. Dwivedi, and S. Jalan, Optimized evolution of networks for principal eigenvector localization, Phys. Rev. E {\bf 96} , 022312 (2017).

\bibitem{evec_localization_2020} P. Pradhan and S. Jalan, From Spectra to Localized Networks: A Reverse Engineering Approach, IEEE Transactions on Network Science and Engineering {\bf 7}(4), (2020).

\bibitem{evec_loc_jaccard_2021} D. Mohapatra, Jaccard Guided Perturbation for Eigenvector Localization, New Generation Computing {\bf 39}, 159-179 (2021).

\bibitem{neurons_localization_2014} R. Chaudhuri, A. Bernacchia, and X. Wang, A diversity of localized timescales in network activity, eLIFE {\bf 3}, e01239 (2014).

\bibitem{evec_loc_temporal_network_2017} D. Taylor, S. A. Myres, Aaron Clauset, M. A. Porter, and P. J. Mucha, Eigenvector-based Centrality Measures for 
Temporal networks, SIAM Multiscale Model. Simul., vol. 15, No. 1, pp. 537-574, (2017).

\bibitem{newman_book_2010} M. E. J. Newman, \emph{Networks: An Introduction}, Oxford University Press, (2010). 

\bibitem{machine_learning_2009} Ravindran Kannan and Santosh Vempala, Spectral Algorithms, Journal Foundations and Trends in Theoretical Computer Science, vol 4, Issue 3-4, pp. 157-288, (2009).

\bibitem{deg_pev_corr_2012} C. Li, H. Wang and P. V. Mieghem, Degree and Principal Eigenvectors in Complex, Networks, Networking, Lecture Notes in Computer Science, vol 7289, 149 (2012).

\bibitem{message_passing_2017} G. Tim\'{a}r, R. A. da Costa, S. N. Dorogovtsev, and J. F. F. Mendes, Nonbacktracking expansion of finite graphs, Phys. Rev. E {\bf 95}, 042322 (2017).

\bibitem{sensitivity_dynamics_2017} T. Nishikawa, J. Sun, and A. E. Motter, Sensitive Dependence of Optimal Network Dynamics on Network Structure, Phys. Rev. X {\bf 7}, 041044 (2017).

\bibitem{community_2010} F. Slanina, and Z. Konopasek, Eigenvector Localization as a tool to study small communities in Online Social Networks, World Scientific, Advances in Complex Systems, (2010).

\bibitem{community_2015} Think locally, act locally: Detection of small, medium-sized, and large communities in large networks, Phys. Rev. E {\bf 91}, 012821 (2015).

\bibitem{mux_community_2017} D. Taylor, R. S. Caceres, and P. J. Mucha, Super-Resolution Community Detection for Layer-Aggregated Multilayer Networks, Phys. Rev X {\bf 7}, 031056 (2017).

\bibitem{anderson_loc_1985} P. A. Lee and T. V. Ramakrishnam, Disordered electronic systems {\it Rev. Mod. Phys.} \textbf{57}, 287 (1985).

\bibitem{loc_math_1_2010} Y. Dekel, J. R. Lee, N. Linial, Eigenvectors of Random Graphs: Nodal Domains, Random Structures \& Algorithms {\bf 39}, 39 (2011).

\bibitem{loc_math_2_2013} C. Bordenave and A. Guionnet, Localization and delocalization of eigenvectors for heavy-tailed random matrices {\it Probability Theory and Related Fields} \textbf{157}, 885-953 (2013).

\bibitem{loc_math_3_2003} X. Liu, G. Strang, and S. Ott, Localized eigenvectors from widely spaced matrix modifications,  SIAM J. Discrete Math., 
vol. 16, No. 3, pp. 479-498, (2003).

\bibitem{approx_algo_2015} David F. Gleich, Michael W. Mahoney, Using Local Spectral Methods to Robustify Graph-Based Learning Algorithms, KDD’15, Sydney, NSW, Australia, 10-13, (2015).

\bibitem{machine_learning_loc_2016} Pan Zhang, Robust Spectral Detection of Global Structures in the Data by Learning a Regularization, 29th Conference on Neural Information Processing Systems (NIPS), Barcelona, Spain (2016).

\bibitem{loc_invariant_subspace_2011} C. V\"{o}mel, and B. N. Parlett, Detecting localization in an invariant subspace, Society for industrial and Applied Mathematics (SIAM) {\bf 33}, 3447 (2011).

\bibitem{anderson_loc_linear_alg_1999}U. Elsner, V. Mehrmann, F. Milde, R. A. Romer, and M. Schreiber, The Anderson model of localization: A challenge for modern eigenvalue methods, SIAM J. Sci. Comput., vol. 20, No. 6, pp. 2089-2102, (1999).

\bibitem{localization_in_mat_2016} M. Benzi, Localization in Matrix Computations: Theory and Applications, Springer International Publishing AG, (2016).

\bibitem{dynamic_reconfig_2011} D. S. Bassett et al., Dynamic reconfiguration of human brain networks during learning, PNAS \textbf{108}, 7641 (2011). 

\bibitem{rev_Strogatz_2001} S. H. Strogatz, Exploring complex networks, Nature \textbf{410}, 268 (2001).

%--------------------------------------------
%\bibitem{power_grid} A. E. Motter et al., Spontaneous synchrony in power-grid networks, Nat. Phys. 9, 191–197 (2013).

%\bibitem{organic_chemistry} B. A. Grzybowski et al., The ‘wired’ universe of organic chemistry, Nat. Chem. 1, 2009.

%\bibitem{robust_classification} Daniel Grady et al., Robust classification of salient links in complex networks. Nature Commun. 864, 2012.

%\bibitem{cosmic_web} B. C. Coutinho et al., The Network Behind the Cosmic Web, arXiv: 1604.03236v2, 13 Apr 2016.

%\bibitem{brain} Alison Abbott, Human Brain Project votes for leadership change, Nature News, 04 March 2015.

\bibitem{quantum_internet} Suzanne van Dam, How can we speed up the quantum internet? \url{https://blog.qutech.nl/2018/11/26/how-can-we-speed-up-the-quantum-internet/}

\bibitem{GPS_network} https://www.businessinsider.com/spacex-watch-live-starlink-internet-satellites-rocket-launch-2019-5

\bibitem{COVID-19} Jin Wu, Weiyi Cai, Derek Watkins and James Glanz, How the Virus Got Out, The New York Times, March 22, 2020.
%--------------------------------------------------------


\bibitem{miegham_book2011} P. V. Mieghem, \emph{Graph Spectra for Complex Networks}, Cambridge University Press, (2011).

\bibitem{search_space} F. Harary and E. Palmer, Graphical Enumeration (Academic Press, New York, 1973).

\bibitem{MZN2014} T. Martin, X. Zhang, and M. E. J. Newman, Localization and Centrality in Networks, Physical Review E {\bf 90}, 052808 (2014).

\bibitem{Goltsev_prl2012} A. V. Goltsev, S. N. Dorogovtsev, J. G. Oliveira, and J. F. F. Mendes, Localization and spreading of diseases in complex networks, Phys. Rev. Lett. \textbf{109}, 128702 (2012).

% \bibitem{evolution_of_networks_2002} S. N. Dorogovtsev, and J. F. F. Mendes, Evolution of networks, Advances in Physics, Vol. 51, No. 4, 1079-1187, (2002).

%\bibitem{Allesina_nat2015} S. Suweis, J. Grilli, J. R. Banavar, S. Allesina, A. Maritan, Effect of localization on the stability of mutualistic ecological networks, Nat. Commun. \textbf{6}, 10179 (2015).

%\bibitem{brain_networks2013} P. Moretti and M. A. Mu\~{n}oz, Griffiths phases and the stretching of criticality in brain networks, {\it Nat. Commun.} \textbf{4}, 2521 (2013).

%\bibitem{hub_loc_diffusion_2017} D. Shorten and G. Nitschke, The Two Regimes of Neutral Evolution: Localization on Hubs and Delocalized Diffusion, Springer International Publishing, Evo Applications, Part I, LNCS 10199, pp. 310-325, (2017).

%\bibitem{loc_dynamics_2013} G\'{e}za \'{O}dor, Spectral analysis and slow spreading dynamics on complex networks, Phys. Rev. E {\bf 88}, 032109 (2013).

\bibitem{low_order_evec_2013} M. Cucuringu, V. D. Blondel, and P. V. Dooren, Extracting spatial information from networks with low-order eigenvectors
Phys. Rev. E {\bf 87}, 032803 (2013).

\bibitem{leverage_score_2011} M. Cucuringu and M. W. Mahoney, Localization on low-order eigenvectors of data matrices, arxiv:1109.1355v1, (2011).

%\bibitem{crypt_loc_2012} D. Mavroeidis, L. Batina, T. V. Laarhoven, E. Marchiori, PCA, eigenvector localization and clustering for side-channel attacks on cryptographic hardware devices, ECML PKDD'12 Proceedings, Springer-Verlag Berlin, Heidelberg, vol Part I, pp. 253-268, (2012).

%\bibitem{spectral_clustering_2014} M. Luci\'{n}ska, A Spectral Clustering Algorithm Based on Eigenvector Localization, ICAISC 2014, Part II, LNAI 8468, pp. 749-759, (2014). 

%\bibitem{google_matrix_2013} L. Ermann, K. M. Frahm, and D. L. Shepelyansky, Spectral properties of Google matrix of Wikipedia and other networks, Eur. Phys. J. B {\bf 86}, 193 (2013).

%\bibitem{KK2015} T. Kawamoto and Y. Kabashima, Limitations in the spectral method for graph partitioning: Detectability threshold and localization of eigenvectors, Phys. Rev. E {\bf 91}, 062803 (2015).

%\bibitem{mux_loc_2017} S. Jalan and P. Pradhan, Localization of multiplex networks by the optimized single-layer rewiring, arXiv:1712.04829v1, (2017).

\bibitem{nonlinear_opt_2014}A. Beck, \emph{Introduction to nonlinear optimization Theory, Algorithms, and Applications with MATLAB}, Society for Industrial and Applied Mathematics, (2014). 

%\bibitem{spectral_radius_2011} P. V. Mieghem, D. Stevanovi\'{c}, F. Kuipers, C. Li, R van de Bovenkamp, D. Liu and H. Wang, Decreasing the spectral radius of a graph by link removals, {\it Phys. Rev. E} {\bf 84}, 016101 (2011).

\bibitem{deloc_pev} L. V. Tran, V. H. Vu, and K. Wang, Sparse Random Graphs: Eigenvalues and Eigenvectors, Random Structures \& Algorithms {\bf 42}, 110 (2013). 

\bibitem{bound_pev_entries_2000} B. Papendieck, P. Recht, On maximal entries in the principal eigenvector of graphs, Elsevier Linear Algebra and its Applications {\bf 310}, 129 (2000). 

\bibitem{algorithms_cormen_2009} T. H. Cormen, C. E. Leiserson, R. L. Rivest, and C. Stein. \emph {Introduction to Algorithms, $3^{rd}$ ed.}, MIT Press Cambridge, (2009).

\bibitem{simuated_annealing_1983} S. Kirkpatrick, C. D. Gelatt, M.P. Vecchi, Optimization by Simulated Annealing, Science \textbf{220}, 4598 (1983).


\end{thebibliography}
\end{document}